\documentclass[preprint]{elsarticle}

\usepackage{amsmath,epsfig}
\usepackage{amssymb}
\usepackage{amsthm}
\usepackage{fullpage}
\usepackage{listings}
\usepackage[ruled,vlined,linesnumbered]{algorithm2e}
\usepackage{qtree}
\usepackage{subfigure}
\usepackage{url}

\usepackage[usenames,dvipsnames]{pstricks}
\usepackage{epsfig}
\usepackage{pst-grad} 
\usepackage{pst-plot} 

\newtheorem{theorem}{Theorem}[section]
\newtheorem{lemma}{Lemma}[section]
\newtheorem{definition}{Definition}[section]

\newtheorem{fact}{Fact}[section]

\newcommand{\qedsymb}{\hfill{\rule{2mm}{2mm}}}
\newenvironment{proofsketch}{\begin{trivlist}
\item[\hspace{\labelsep}{\noindent Proof Sketch: }]
}{\qedsymb\end{trivlist}}

\newcommand{\remove}[1]{}

\DeclareMathOperator{\sa}{\mathit{SA}}
\DeclareMathOperator{\rank}{\mathit{Rank}}
\DeclareMathOperator{\lcp}{\mathit{LCP}}

\DeclareMathOperator{\lr}{\mathit{LR}}
\DeclareMathOperator{\llr}{\mathit{LLR}}
\DeclareMathOperator{\llrs}{\mathit{LLRS}}
\DeclareMathOperator{\lrs}{\mathit{LRS}}

\begin{document}

\begin{frontmatter}

  \title{On Longest Repeat Queries\tnoteref{t1}}

 \tnotetext[t1]{Author names are in alphabetical order. 
}

\author[a1]{Atalay Mert \.{I}leri}
\address[a1]{Massachusetts Institute of Technology, USA}
\ead{atalay@mit.edu}

\author[a2]{M. O\u{g}uzhan  K\"{u}lekci}
\address[a2]{Istanbul Medipol University, Turkey}
\ead{okulekci@medipol.edu.tr}

\author[a3]{Bojian Xu\corref{c1}} 
\address[a3]{
Eastern Washington University, USA}
\ead{bojianxu@ewu.edu} \cortext[c1]{Corresponding author. Mailing address: 319F CEB,
  Eastern Washington University,
  Cheney, WA 99004, USA. Phone:
  +1 (509) 359-2817. Fax: +1 (509) 359-2215. 
}

\begin{abstract} 
  Repeat finding in strings has important applications in subfields
  such as computational biology. Surprisingly, all prior work on
  repeat finding did not consider the constraint on the locality of
  repeats. In this paper, we propose and study the problem of finding
  longest repetitive substrings covering particular string positions.
  We propose an $O(n)$ time and space algorithm for finding the
  longest repeat covering every position of a string of size $n$. Our
  work is optimal since the reading and the storage of an input string
  of size $n$ takes $O(n)$ time and space.  Because any
  substring of a repeat is also a repeat, our solution to longest
  repeat queries effectively provides a ``stabbing'' tool for
  practitioners for finding most of the repeats that cover particular string
  positions.
 \end{abstract}

 \begin{keyword}
   information retrieval\sep string processing \sep repeats \sep
   regularities \sep repetitive structures
 \end{keyword}

\end{frontmatter}

\section{Introduction}

Repetitive structures and regularities finding in genomes and proteins
is important as these structures play important roles in the 
biological functions of genomes and proteins~\cite{Gus97}.  It is
well known that overall about one-third of the whole human genome
consists of repeated subsequences~\cite{McC1993}; about 10--25\% of
all known proteins have some form of repetitive
structures~\cite{LW06}. In addition, a number of significant problems
in molecular sequence analysis can be reduced to repeat
finding~\cite{Mar83}. Another motivation for finding repeats is to
compress the DNA sequences, which is known as one of the most
challenging tasks in the data compression field. DNA sequences consist
only of symbols from {\tt \{ACGT\}} and therefore can be represented
by two bits per character. Standard compressors such as {\tt gzip} and
{\tt bzip} usually use more than two bits per character and therefore
cannot reach good compression. Many modern genomic sequence data
compression techniques highly rely on the repeat finding in the
sequences~\cite{MR04,BF05}.

The notion of maximal repeat and super maximal
repeat~\cite{Gus97,KVX-BIBM2010,KVX-tcbb2012,BBO-spire2012} captures
all the repeats of the whole string in a space-efficient manner, but
it does not track the locality of each repeat and thus can not support
the finding of repeats that cover a particular string position.
In this paper, we propose and study the problem of finding longest
repetitive substrings covering any particular string positions.
Because any substring of a repeat is also a repeat, the solution to
longest repeat queries effectively provides a ``stabbing'' tool for
practitioners for finding most of the repeats that cover particular
string positions.
\remove{
The recent study of finding shortest unique substrings (SUS) covering
particular string positions~\cite{PWY-ICDE2013,
TIBT2014,IKX-CPM2014} may help find some repeats
covering particular string positions, because any extension of a SUS
must be a repeat.  However, such trivial reduction does not guarantee
the finding of the longest repeat that covers an arbitrary string
position. In fact, it is not yet well understood how to reduce longest
repeat finding to SUS finding.
}

In this paper, we propose an $O(n)$ time and space algorithm that can
find the \emph{leftmost} longest repeat of every string
position. We view our solution to be optimal in both time and space, 
because one has to spend $\Omega(n)$ time and space to read
and store the input string.


\section{Preliminary}
\label{sec:prelim}
We consider a {\bf string} $S[1\ldots n]$,
 where each character $S[i]$ is
drawn from an alphabet $\Sigma=\{1,2,\ldots, \sigma\}$.
A {\bf substring} $S[i\ldots j]$
of $S$ represents $S[i]S[i+1]\ldots S[j]$ if $1\leq i\leq j \leq n$,
and is an empty string if $i>j$.
String $S[i'\ldots j']$ is a {\bf proper substring} of another string
$S[i\ldots j]$ if $i\leq i' \leq j' \leq j$ and $j'-i' < j-i$. 
%
%
The {\bf length} of a non-empty substring $S[i\ldots j]$, denoted as
$|S[i\ldots j]|$, is $j-i+1$. We define the length of an empty string
as zero. 
A {\bf prefix} of $S$ is a substring $S[1\ldots i]$
for some $i$, $1\leq i\leq n$. 
A {\bf proper prefix} $S[1\ldots i]$ is a prefix of $S$ where $i <
n$.
A {\bf suffix} of $S$ is a substring
$S[i\ldots n]$ for some $i$, $1\leq i\leq n$.  
A {\bf proper suffix} $S[i\ldots n]$ is a suffix of $S$ where $i >
1$.
We say the character $S[i]$ occupies the string {\bf position} $i$.
We say the substring $S[i\ldots j]$ {\bf covers} the $k$th position of
$S$, if $i\leq k \leq j$.  
For two strings $A$ and $B$, we write ${\bf A=B}$ (and say $A$ is {\bf
  equal} to $B$), if $|A|= |B|$ and $A[i]=B[i]$ for 
$i=1,2,\ldots, |A|$.  
%
We say $A$ is lexicographically smaller than $B$,
denoted as ${\bf A < B}$, if (1) $A$ is a proper prefix of $B$, or (2)
$A[1] < B[1]$, or (3) there exists an integer $k > 1$ such that
$A[i]=B[i]$ for all $1\leq i \leq k-1$ but $A[k] < B[k]$.
A substring
$S[i\ldots j]$ of $S$ is {\bf unique}, if there does not exist
another substring $S[i'\ldots j']$ of $S$, such that 
$S[i\ldots j] = S[i'\ldots j']$ but $i\neq i'$. 
A substring is a {\bf repeat} if it is not unique.
A character $S[i]$ is a {\bf singleton}, if it appears only once in
$S$.

\begin{definition}
\label{def:lr}
For a particular string position $k\in \{1,2,\ldots, n\}$,  
the {\bf longest repeat (LR) covering position} ${\bf k}$, denoted
as $\mathit{\bf \lr}_{\bf k}$, is 
a repeat substring $S[i\ldots j]$, such that: (1) $i\leq k \leq j$, and 
(2) there does not exist  another repeat substring $S[i'\ldots j']$, such
that $i'\leq k \leq j'$ and $j'-i' > j-i$. 
\end{definition}

\begin{definition}
\label{def:llr}
For a particular string position $k\in \{1,2,\ldots, n\}$, the
{\bf left-bounded longest repeat (LLR) starting at position $k$},
denoted as ${\mathit{\bf LLR}_{\bf k}}$, is a repeat $S[k\ldots j]$,
such that either $j=n$ or $S[k\ldots j+1]$ is unique. 
\end{definition}

Obviously, 
for any string position $k$, if $S[k]$ is not a singleton, both
$\lr_k$ and $\llr_k$
must exist, because at least $S[k]$ itself is a repeat. 
Further, there might be multiple choices for $\lr_k$. For
example, if $S={\tt abcabcddbca}$, then $\lr_2$ can be either
$S[1\ldots 3]={\tt
  abc}$ or $S[2\ldots 4]={\tt bca}$.
However, if $\llr_k$ does exist, it
must have only one choice, because $k$ is a fixed string position and
the length of $\llr_k$ must be as long as possible.

The {\bf suffix array} $\sa[1\ldots n]$ of the string $S$ is a
permutation of $\{1,2,\ldots, n\}$, such that for any $i$ and $j$,
$1\leq i < j \leq n$, we have $S[\sa[i]\ldots n] < S[\sa[j]\ldots n]$.
That is, $\sa[i]$ is the starting position of the $i$th suffix in
the sorted order of all the suffixes of $S$.
The {\bf rank array} $\rank[1\ldots n]$ is the inverse of the suffix
array. That is, $\rank[i]=j$ iff $\sa[j]=i$. 
The {\bf longest common prefix (lcp) array} $\lcp[1\ldots n+1]$ is an
array of $n+1$ integers, such that for $i=2,3,\ldots, n$, $\lcp[i]$ is
the length of the lcp of the two suffixes $S[\sa[i-1]\ldots n]$ and
$S[\sa[i]\ldots n]$. We set $\lcp[1]=\lcp[n+1]=0$.  In the literature,
the lcp array is often defined as an array of $n$ integers. We include
an extra zero at $\lcp[n+1]$ is only to simplify the description 
of our upcoming
algorithms.  
Table~\ref{tab:suflcp}
in the appendix 
shows the suffix array and the lcp
array of the example string {\tt mississippi}.

The next Lemma~\ref{lem:llr} shows that, by using the rank array and
the lcp array of the string $S$, it is easy to calculate any $\llr_i$ if
it exists or to detect the fact that it does not exist.

\begin{lemma}
\label{lem:llr}
For $i=1,2,\ldots,n$: 
$$
\llr_i = 
\left \{
\begin{array}{lll}
S[i\ldots i + L_i-1] & , & \textrm{\ \ \ if \ \ } L_i > 0\\
\textit{does not exist} & , & \textrm{\ \ \ if \ \ }L_i = 0
\end{array}
\right.
$$
where $L_i = \max\{\lcp[\rank[i]],\lcp[\rank[i]+1]\}$.
\end{lemma}

\begin{proof}
  Note that $L_i$ is the length of the lcp between the suffix
  $S[i\ldots n]$ and any other suffix of $S$.  If $L_i > 0$, it means
  substring $S[i\ldots L_i-1]$ is the lcp among $S[i\ldots n]$ and any
  other suffix of $S$. So $S[i\ldots L_i-1]$ is $\llr_i$.  Otherwise
  ($L_i = 0$), the letter $S[i]$ is a singleton, so $\llr_i$ does not
  exist.
\end{proof}

\section{Longest repeat finding for one position}   
\label{sec:one}
In this section, we want to find $\lr_k$ for a given string position
$k$, using $O(n)$ time and space.  We present the solution to this
setting
here in 
case the practitioners have only a smaller number of 
string positions, for which they want to find the longest repeats, 
and thus this light-weighted solution will suffice.
We will start with finding the leftmost $\lr_k$ if the string position
$k$ is covered by multiple LRs. In the end of the section, we will
show a trivial extension to find all LRs covering position $k$ with
the same time and space complexities, if $k$ has multiple LRs.

\begin{lemma}
\label{lem:lr-llr}
Every LR is an LLR.
\end{lemma}

\begin{proof}
Assume that $\lr_k=S[i\ldots j]$ is not an LLR. Note that $S[i\ldots
j]$ is a repeat starting from position $i$. If $S[i\ldots j]$ is not
an LLR, it means $S[i\ldots j]$ can be extend to some position
$j' > j$, so that $S[i\ldots j']$ is still a repeat and  also covers
position $k$. That says, $|S[i\ldots j']| > |S[i\ldots j]|$.
However, the contradiction is that $S[i\ldots j]$ is already the longest repeat
covering position $k$. 
\end{proof}

\begin{lemma}
\label{lem:llr-cover}
For any three string positions $i$, $j$, and $k$, $1\leq i < j \leq k
\leq n$: if
$\llr_j$ does not exist or does not cover position $k$, $\llr_i$ does
not exist or does not cover position $k$ either.
\end{lemma}

\begin{proof}
  (1) If $\llr_j$ does not exist, then $S[j]$ is a singleton.  If
  $\llr_i$ does exist and covers position $k$, then $\llr_i$ also covers
 position $j$, which yields a contradiction that 
  the substring $\llr_i$ includes the singleton $S[j]$ but is a repeat. 
  (2) If $\llr_j=S[j\ldots t]$ does exist but does not cover position
  $k$, then $S[j\ldots t+1]$ is unique and $t+1 \leq k$. If $\llr_i$
  exists and covers position $k$, say $\llr_i = S[i\ldots r]$, 
  $r\geq k$, it means $S[j\ldots t+1]$ is a substring of a repeat
  $\llr_i=S[i\ldots r]$, because $i<j<t+1 \leq r$, so $S[j\ldots t+1]$
  is also a repeat. This contradicts to the fact that
  $S[j\ldots t+1]$ is unique. So   $\llr_i$ does not exist or does not
  cover position $k$. 
\end{proof}

The idea behind the algorithm for finding the LR covering a given
position is straightforward.  Algorithm~\ref{algo:one} shows the
pseudocode, where the found LR is returned as a tuple
$\langle start,length\rangle$, representing the starting position and the length of
the LR, respectively. 
If the LR that is being searched for does not
exist, $\langle -1,0\rangle $ is returned by Algorithm~\ref{algo:one}.
 We know that any longest repeat covering
position $k$ must be an LLR (Lemma~\ref{lem:lr-llr}), starting between
indexes $1$ to $k$ inclusive.
What we need to do is to simply compute every individual of $\llr_1
\dots \llr_k$ using Lemma~\ref{lem:llr} and check whether it covers
position $k$ or not.  We will just choose the longest LLR that covers
position $k$ and resolve the tie by picking the leftmost one if $k$ is
covered by multiple LRs (Line~\ref{line:tie}).
Due to Lemma~\ref{lem:llr-cover}, a practical speedup is possible via
an early stop (Line~\ref{line:break}) by computing and checking from
$\llr_k$ down to $\llr_1$ (Line~\ref{line:for}). 
   
\begin{algorithm}[t]
{\small
  \caption{Find $\lr_k$. Return the leftmost one if $k$ has multiple LRs.}
\label{algo:one}
\KwIn{The position index $k$, and the rank array and 
      the lcp array of the string $S$} 
\KwOut{$\lr_k$ or find no such LR. The leftmost one will be returned if $k$ has multiple
LRs.}

\smallskip 

$start \leftarrow -1$; $length \leftarrow 0$ 
\tcp*{start position and length of $\lr_k$}

\For{$i = k$ down to $1$\label{line:for}}{
  $L \leftarrow \max\{\lcp[\rank[i]],\lcp[\rank[i]+1]\}$\tcp*{Length
    of $\llr_i$}

  \If(\tcp*[f]{$\llr_i$ does not exist or does not cover $k$.}){$L=0$ or $i+L-1<k$}{break\tcp*{Early stop}\label{line:break}}
  \ElseIf(\tcp*[f]{Tie is resolved by picking the leftmost one.}){$L\geq length$\label{line:tie}}{$start\leftarrow i$; $length\leftarrow L$\;}
}
Print $\lr_k\leftarrow \langle start,length\rangle$\;
}
\end{algorithm}

\begin{lemma}
\label{lem:one}
Given the rank array and the lcp
array of the string $S$,
for any position $k$ in the string $S$, Algorithm~\ref{algo:one} 
can find $\lr_k$ or the
fact that it does not exist, using
$O(k)$ time and $O(n)$ space.  If there are multiple candidates for
$\lr_k$, the leftmost one is returned.
 \end{lemma}

 \begin{proof}
   The algorithm clearly has no more than $k$ steps and each step
   takes $O(1)$ time, so it costs a total of $O(k)$ time. The space cost is primarily
  from the rank array and the lcp array, which altogether is $O(n)$,
   assuming each integer in these arrays costs a constant number of
   bytes.  

 \end{proof}

\begin{theorem}
\label{thm:one}
For any position $k$ in the string $S$, we can find $\lr_k$ or the
fact that it does not exist, using
$O(n)$ time and  space.  If there are multiple candidates for
$\lr_k$, the leftmost one is returned.
\end{theorem}

\begin{proof}
  The suffix array of $S$ can be constructed by existing algorithms
  using $O(n)$ time and space (For ex., \cite{KA-SA2005}). After the
  suffix array is constructed, the rank array can be trivially created
  using $O(n)$ time and space.  We can then use the suffix array and
  the rank array to construct the lcp array using another $O(n)$ time
  and space~\cite{KLAAP01}.  Given the rank array and the lcp array,
  the time cost of Algorithm~\ref{algo:one} is $O(k)$
  (Lemma~\ref{lem:one}). So altogether, we can find $\lr_k$ or the
  fact that it does not exists using $O(n)$ time and space.
  If multiple LRs cover position $k$, the leftmost LR will be returned
  as is guaranteed by Line~\ref{line:tie} of
  Algorithm~\ref{algo:one}. 
\end{proof}


\smallskip 
\noindent
{\bf Extension: Find all LRs covering a given position.} 
It is trivial to extend Algorithm~\ref{algo:one} to find all the LRs
covering any given position $k$ as follows. We can first use
Algorithm~\ref{algo:one} to find the leftmost $\lr_k$. If $\lr_k$ does
exist, then we will start over again to recheck $\llr_k$ down to
$\llr_1$ and return those whose length is equal to the length of
$\lr_k$. Due to Lemma~\ref{lem:llr-cover}, the same early stop as we
have in Algorithm~\ref{algo:one} can be used for a practical
speedup. 
The pseudocode of this
procedure is provided in Algorithm~\ref{algo:one-all} in the appendix, which 
clearly costs an extra $O(k)$ time. Combining
Theorem~\ref{thm:one}, we have:

\begin{theorem}
\label{thm:one-all}
We can find all the LRs covering any given position $k$ using $O(n)$ time
and space.
\end{theorem}

\section{Longest repeat finding for every position}   
\label{sec:every}
In this section, we want to find $\lr_k$ of every position
$k=1,2,\ldots,n$. 
If any position $k$ is covered by multiple LRs, the
leftmost one will be returned.  
A natural solution is to iteratively use Algorithm~\ref{algo:one} as a
subroutine to find every $\lr_k$, for $k=1,2,\ldots,n$. However, the
total time cost of this solution will be $O(n)+\sum_{k=1}^n O(k) =
O(n^2)$, where $O(n)$ captures the time cost for the construction of
the rank array and the lcp array and $\sum_{k=1}^n O(k)$ is the total
time cost for the $n$ instances of Algorithm~\ref{algo:one}. We want to
have a solution that costs a total of $O(n)$ time and space, which follows
that the amortized cost for finding each LR is $O(1)$.

\subsection{A conceptual algorithm}
\label{subsec:overall}
We will first calculate $\llr_1,\llr_2,\ldots,\llr_n$ using
Lemma~\ref{lem:llr}, and save the results in an array $\llrs[1\ldots
n]$. Each LLR is represented by a tuple $\langle start,length\rangle$, the starting
position and the length of the LLR. We assign zero as the length of
any non-existing LLR, which does not cover any string position. We
then sort the $\llrs$ array in the descending order of the lengths of
the LLRs, using a stable and linear-time sorting procedure such as the
counting sort.  

\begin{definition}
\label{def:p}
After the $\llrs$ array is stably sorted, let $P_1$
denote the string positions that are covered by $\llrs[1]$, and $P_i$,
$2\leq i\leq n$, denote the string positions that are covered by
$\llrs[i]$ but are not covered by any of $\llrs[1\ldots i-1]$. Let
$|P_i|$ denote the number of string positions belonging to $P_i$.
\end{definition}

Note that any $P_i$, $i\geq 1$, can possibly be empty. 
Our conceptual algorithm will then assign $\llrs[i]$ as the LR of
those string positions belonging to $P_i$, if $P_i$ is not empty, for
$i=1,2,\ldots,n$. We store the LRs that we have calculated in an array
$\lrs[1\ldots n]$ of $\langle start, length\rangle$ tuples, where
$\lrs[i]=\lr_i$ and $\lrs[i].start$ and $\lrs[i].length$ represent the
starting position and length of $\lr_i$.  If $\lr_i$ does not exist,
the tuple $\langle -1,0\rangle$ will be assigned to $\lrs[i]$, which
can be done during the initialization of the $\lrs$ array.  Early stop
can be made when (1) we meet an $\llrs$ array element whose length is
zero, which indicates that all the remaining $\llrs$ array elements
also have lengths of zero; or (2) every string position has had their
LR calculated. Algorithm~\ref{algo:concept} shows the pseudocode of
this conceptual algorithm.


\begin{algorithm}[t]
{\small
  \caption{The conceptual algorithm for finding the leftmost LR for
    every non-singleton string position of $S$.}
\label{algo:concept}
\KwIn{The rank array and the lcp array of the string $S$} 
\KwOut{The leftmost LR covering every non-singleton string position of $S$.}

\smallskip 

\tcc{Calculate the $\llrs$ array using Lemma~\ref{lem:llr}. 
Initialize the $LRS$ array..}
\For{$i=1,2,\ldots,n$}{
 $\llrs[i] \leftarrow (i,\max\{\lcp[\rank[i]], \lcp[\rank[i]+1]\})$ 
  \tcp*{$\llr_i$, in the format of $\langle start, length\rangle$}
  $\lrs[i] \leftarrow \langle -1,0\rangle$ \tcp*{$\lr_i$, in the format of $\langle start, length \rangle$}
}

\smallskip 

Stably sort $\llrs[1\ldots n]$ in the descending order of its second
dimension \label{line:sort}\tcp*{e.g.: counting sort.} 

\bigskip 

\tcc{Find the leftmost LR for every position}

$count \leftarrow 0$ \tcp*{The number of non-singleton string positions that have their
  LRs calculated.} 

\For{$i=1,2,\ldots,n$}{
  \lIf{$count = n$ \emph{or} $\llrs[i].length = 0$}{
    break \tcp*{Early stop}
  }
  \lIf{$|P_i| = 0$}{continue\;}

  \lForEach{$k\in P_i$}{
    $\lrs[k] \leftarrow \llrs[i]$  
    \tcp*{Calculate the LRs of the positions belonging to $P_i$.}
  }
  $count \leftarrow count + |P_i|$\;
 }

 \Return{$\lrs[1\ldots n]$}

}
\end{algorithm}


\begin{lemma}
\label{lem:concept}
Algorithm~\ref{algo:concept} finds the LR for every position that does not
contain a singleton. It finds the leftmost LR if any
position is covered by multiple LRs.
\end{lemma}
	    
\begin{proof}
\remove{
  We prove the lemma by induction. 
  We first assume every string
  position has at most one LR and will deal with the case where one
  position can be possibly covered by multiple LRs in the end of the
  proof. 
}
The proof of the lemma is obvious. 
Recall that every LR must be an LLR (Lemma~\ref{lem:lr-llr}) and we
process all LLRs in descending order of their lengths.
For $i=1,2,\ldots,n$, if $P_i$ is not empty, 
then for each position in $P_i$, 
the  substring $\llrs[i]$ is the
longest LLR that covers that position,
i.e., $\llrs[i]$ is the LR of
that position. 
\remove{
  Base case: We
  first observe that $\llrs[1]$ is the longest repeat among
  $\llr_1,\llr_2,\ldots,\llr_n$, unless every string position in $S$
  contains a singleton. Therefore,  $\llrs[1]$ is the LR of every position in $P_1$.
  
  Let's assume the conceptual algorithm produces correct LRs for
  positions that are covered by $\llrs[1\ldots i-1]$.  We will
  then show that it will produce correct LRs for positions that
  are covered by $\llrs[i]$ but are not covered by
  $\llrs[1\ldots i-1]$. For the purpose of proof based on
  contradiction, we assume that $\llrs[i]$ is
  not an LR for all those positions that are not covered by $\llrs[1\ldots
  i-1]$ but are covered by $\llrs[i]$, then this implies
there exists at lease one position $k$ of those positions, such that 
$k$ is covered by an LR that is longer than $\llrs[i]$.
From Lemma~\ref{lem:lr-llr}, we know that longer LR is also
  an LLR. Because we are processing LLRs in non-increasing order, 
  that longer LLR should appear before $\llrs[i]$. This contradicts
  with the fact that position $k$ is not covered by
  $\llrs[1\ldots i-1]$. So, $\llrs[i]$ is the LR of those
  positions that are covered by $\llrs[i]$ but are not covered
  by $\llrs[1\ldots i-1]$. 
}
In the case where any position in $P_i$
  has multiple LRs, $\llrs[i]$ must be the leftmost LR because
  of the stable sorting of the $\llrs$ array. 
\end{proof}

\subsection{High-level strategy for a fast implementation}
\label{subsec:table}
The challenge is to implement the conceptual algorithm
(Algorithm~\ref{algo:concept}) efficiently. Our goal is to use $O(n)$
time and space only, which is optimal since we have to spend $O(n)$
time and space to report all the LRs of all the $n$ distinct string
positions.
We start with some property of each $P_i$ (Definition~\ref{def:p}).
Recall that, in Algorithm~\ref{algo:concept}, we process all the LLRs
in the descending order of their lengths, and also all LLRs start from
distinct string positions. Therefore, after the $\llrs$ array is
sorted (Line~\ref{line:sort}, Algorithm~\ref{algo:concept}),
none of
$\llrs[1\ldots i-1]$ can be a substring of $\llrs[i]$,  for any $i\geq
2$. This yields the following fact.


\begin{fact}
\label{fact:p}
Every non-empty $P_i$, $i\geq 1$, is a continuous chunk of
string positions, i.e., every non-empty $P_i$ is an integer range $[s_i,e_i]$,
where $s_i$ and $e_i$ are the starting
and ending string positions of $P_i$. 
\end{fact}
In the case where $P_i$ is empty, we set $s_i = e_i = -1$.
%
In order to achieve an overall $O(n)$-time implementation of
Algorithm~\ref{algo:concept}, we need a mechanism that can quickly
find $s_i$ using $O(1)$ time when processing each $\llrs[i]$. Then, if
$s_i\neq -1$, due to Fact~\ref{fact:p}, we can just linearly walk from
string position $s_i$ through the position $e_i$, which is either the
right boundary of $\llrs[i]$ or a string position whose next
neighboring position has had its LR calculated, whichever one
is reached first.  We will then set the LR of each visited position
during the walk to be $\llrs[i]$, achieving an overall $O(n)$ time
implementation of Algorithm~\ref{algo:concept}.

When we process a non-empty $\llrs[i]$ and calculate its $s_i$, there
are two cases.  Case 1: The string position $\llrs[i].start$ has not
had its LR calculated, then obviously $s_i = \llrs[i].start$.  Case 2:
The string position $\llrs[i].start$ has already had its LR
calculated, then it is either $s_i > \llrs[i].start$ (if $P_i$ is not
empty) or $s_i = -1$ (if $P_i$ is empty). In this case, it will not be
efficient to find $s_i$ by simply walking from $\llrs[i].start$ toward
$e_i$ until we reach $e_i$ or a string position whose LR has not been
calculated. It is not immediately clear how to calculate
$s_i$ using $O(1)$ time. This leads to the design of our following
mechanism that enables us to calculate every $s_i$ in Case 2 using
$O(1)$ time.

\subsection{The two-table system: the $ptr$ and $next$ arrays}
\remove{
\begin{definition}
  \label{def:alpha-k}
  For any string position $k$ that has had its LR calculated during
  the run of Algorithm~\ref{algo:concept}, we use $\alpha_k$ to denote
  the first string position that is after $k$ and has not had its LR
  calculated, if such string position exists; otherwise, we set
  $\alpha_k = n+1$, where $n$ is the length of the string $S$.
\end{definition}
}

Our mechanism is built upon two integer arrays, $ptr[1\ldots n]$ and
$next[1\ldots n]$.  We
update the two arrays online when we process the sorted $\llrs$ array
elements in the calculation of the LR of every string position.
\emph{Ideally}, we want to maintain these two arrays, such that for
any string position $k$ that has had its LR calculated,
$next\bigl[ptr[k]\bigr]$ is either the next after-$k$ string position
whose LR is not calculated yet or $n+1$ if no such after-$k$ string
position exists.  Then, when we process a particular non-empty
$\llrs[i]$, if the string position $\llrs[i].start$ has had its LR
calculated, we can either directly get $s_i$ or find the fact that all
string positions covered by $\llrs[i]$ have had their LRs calculated,
by comparing $next\Bigl[ptr\bigl[\llrs[i].start\bigr]\Bigr]$ and
$\llrs[i].start+\llrs[i].length-1$ (the right boundary of $\llrs[i]$).
However, it is not clear how to achieve such an ideal maintenance of
the $ptr$ and $next$ arrays in a time-efficient manner. This motivates
us to maintain these two arrays \emph{approximately}, which is to
maintain the following invariance. We will show later that such
approximate maintenance of the $ptr$ and $next$ arrays can still help
calculate every $s_i$ using $O(1)$ time.

\subsubsection{Invariance.}
We initialize every element of both $ptr$ and $next$ arrays to be $-1$.
Recall that after the $\llrs$ array
is sorted in descending order of the lengths of the LLRs, we process
every $\llrs[i]$, for $i=1,2,\ldots,n$. After we have finished the
processing of $\llrs[1\ldots i-1]$, for any $i\geq 2$, we want to
maintain the following invariance for the $ptr$ and $next$ arrays when processing $\llrs[i]$.

\begin{enumerate}
\item If $\llrs[i].start$ has already had its LR calculated but $|P_i|>0$, then:
  $$next\Bigl[ptr\bigl[\llrs[i].start\bigr]\Bigr]  =  s_i$$

\item If $|P_i|=0$ but $\llrs[i].length > 0$ (i.e.: $\llrs[i]$ is not
  empty), then $next\Bigl[ptr\bigl[\llrs[i].start\bigr]\Bigr]$ is
  larger than the index of the right boundary of $\llrs[i]$. That is,
  $$next\Bigl[ptr\bigl[\llrs[i].start\bigr]\Bigr]  > 
  \llrs[i].start + \llrs[j].length-1$$ 
\end{enumerate}

\subsubsection{Using the invariance.}
Recall that when we process a particular non-empty $\llrs[i]$, we want
to calculate $s_i$ quickly. The hard case is when the string position
$\llrs[i].start$ has already had its LR calculated.  Provided with the
above invariance of the $ptr$ and $next$ arrays, when we process a non-empty $\llrs[i]$, we will first check the
value of $ptr\bigl[\llrs[i].start\bigr]$. If it is not equal to $-1$,
the hard case occurs. Then, if
$next\Bigl[ptr\bigl[\llrs[i].start\bigr]\Bigr] \leq \llrs[i].start +
\llrs[i].length -1$ (the right boundary of $\llrs[i]$), we can assert
$s_i =next\Bigl[ptr\bigl[\llrs[i].start\bigr]\Bigr]$; otherwise, we
can assert $P_i$ is empty and thus will simply skip $\llrs[i]$.

\subsection{Maintaining the two-table system.}
\label{subsec:maintain}
In the following, we will first describe how we update the $ptr$ and
$next$ arrays when processing every $\llrs[i]$. In the end, we will
explain why the invariance is maintained using an overall $O(n)$ time.
Remind that the whole algorithm will early stop if $\llrs[i]$ is
empty, so we will only need to describe the algorithmic for
processing a non-empty $\llrs[i]$.  We first initialize every element
in both $ptr$ and $next$ arrays to be $-1$.  We will use the word
\emph{bucket} to denote a maximal and continuous area in the $ptr$
array where all entries share the same positive value.  So initially,
there is no bucket presented in the $ptr$ array.  Because all the
$\llrs$ array elements have been sorted in the descending order of
their lengths, the maintenance of the two-table system will only have
the following five cases to consider (Figure~\ref{fig:cases}).  We use
$left$ and $right$ to denote the indexes of the left and right
boundary of the $\llrs[i]$. That is,

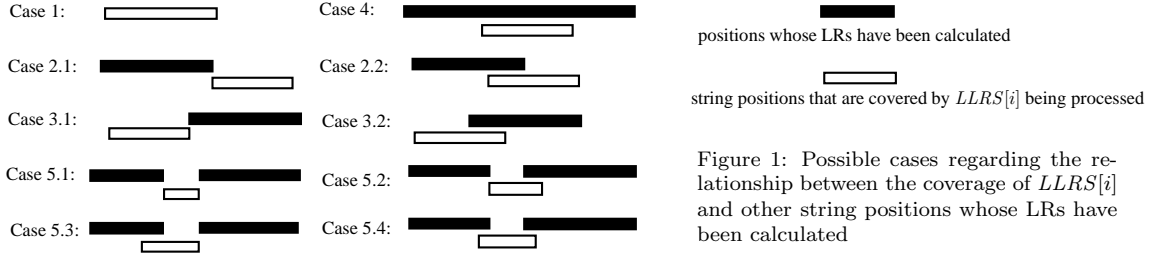
\begin{figure}[t]
  \centering
\begin{minipage}[c]{3.6in}
\scalebox{0.7} 
{
\begin{pspicture}(0,-2.3819044)(21.78787,2.4219043)
\psframe[linewidth=0.04,dimen=outer,fillstyle=solid](4.008242,2.2831345)(1.8482422,2.0431347)
\psframe[linewidth=0.04,dimen=outer,fillstyle=solid,fillcolor=black](3.9282422,1.2831343)(1.7682422,1.0431347)
\psframe[linewidth=0.04,dimen=outer,fillstyle=solid](5.448242,0.96313477)(3.8882422,0.72313434)
\psframe[linewidth=0.04,dimen=outer,fillstyle=solid,fillcolor=black](9.848242,1.3231347)(7.688242,1.0831347)
\psframe[linewidth=0.04,dimen=outer,fillstyle=solid](10.888242,1.0031346)(9.1282425,0.7631345)
\psframe[linewidth=0.04,dimen=outer,fillstyle=solid,fillcolor=black](5.608242,0.28313467)(3.4482422,0.043134566)
\psframe[linewidth=0.04,dimen=outer,fillstyle=solid](3.4882421,0.0031346655)(1.9282422,-0.23686534)
\psframe[linewidth=0.04,dimen=outer,fillstyle=solid,fillcolor=black](10.928242,0.24313462)(8.768242,0.0031346157)
\psframe[linewidth=0.04,dimen=outer,fillstyle=solid](9.488242,-0.07686538)(7.7282424,-0.31686538)
\psframe[linewidth=0.04,dimen=outer,fillstyle=solid,fillcolor=black](5.5869923,-0.79686534)(3.6482422,-1.0219043)
\psframe[linewidth=0.04,dimen=outer,fillstyle=solid](3.6669922,-1.1568655)(2.9682422,-1.3819042)
\psframe[linewidth=0.04,dimen=outer,fillstyle=solid,fillcolor=black](2.9869921,-0.79686534)(1.5682422,-1.0368654)
\psframe[linewidth=0.04,dimen=outer,fillstyle=solid,fillcolor=black](11.968243,-0.71686566)(9.808242,-0.9568653)
\psframe[linewidth=0.04,dimen=outer,fillstyle=solid,fillcolor=black](9.186993,-0.71686566)(7.626992,-0.9419043)
\psframe[linewidth=0.04,dimen=outer,fillstyle=solid](10.186993,-1.0368652)(9.146992,-1.3019043)
\psframe[linewidth=0.04,dimen=outer,fillstyle=solid,fillcolor=black](11.946992,2.3231347)(7.528242,2.0831347)
\psframe[linewidth=0.04,dimen=outer,fillstyle=solid](10.768242,1.9631345)(9.008243,1.7231345)
\usefont{T1}{ptm}{m}{n}
\rput(0.58564454,2.2030957){Case 1:}
\usefont{T1}{ptm}{m}{n}
\rput(6.4656444,2.2430956){Case 4:}
\usefont{T1}{ptm}{m}{n}
\rput(0.6856445,1.1630957){Case 2.1:}
\usefont{T1}{ptm}{m}{n}
\rput(6.6056447,1.1630957){Case 2.2:}
\usefont{T1}{ptm}{m}{n}
\rput(0.6856445,0.1630957){Case 3.1:}
\usefont{T1}{ptm}{m}{n}
\rput(6.6456447,0.123095706){Case 3.2:}
\psframe[linewidth=0.04,dimen=outer,fillstyle=solid,fillcolor=black](5.5469923,-1.7968653)(3.6482422,-2.0219042)
\psframe[linewidth=0.04,dimen=outer,fillstyle=solid](3.6669922,-2.1568654)(2.5469923,-2.3819044)
\psframe[linewidth=0.04,dimen=outer,fillstyle=solid,fillcolor=black](2.9869921,-1.7968653)(1.5682422,-2.0368652)
\psframe[linewidth=0.04,dimen=outer,fillstyle=solid,fillcolor=black](11.968243,-1.7168657)(9.808242,-1.9568653)
\psframe[linewidth=0.04,dimen=outer,fillstyle=solid,fillcolor=black](9.186993,-1.7168657)(7.626992,-1.9419043)
\psframe[linewidth=0.04,dimen=outer,fillstyle=solid](10.066992,-2.0368652)(8.946992,-2.3019042)
\usefont{T1}{ptm}{m}{n}
\rput(0.64564455,-0.8769043){Case 5.1:}
\usefont{T1}{ptm}{m}{n}
\rput(6.6456447,-1.9169043){Case 5.4:}
\usefont{T1}{ptm}{m}{n}
\rput(6.6456447,-0.9969043){Case 5.2:}
\usefont{T1}{ptm}{m}{n}
\rput(0.7256445,-1.9569043){Case 5.3:}
\psframe[linewidth=0.04,dimen=outer,fillstyle=solid](16.906992,1.0580957)(15.506992,0.8180957)
\psframe[linewidth=0.04,dimen=outer,fillstyle=solid,fillcolor=black](16.866993,2.3231347)(15.448242,2.0831347)
\usefont{T1}{ptm}{m}{n}
\rput(16.098047,1.7230957){positions whose LRs have been calculated}
\usefont{T1}{ptm}{m}{n}
\rput(17.298643,0.5230957){string positions that are covered by $\llrs[i]$ being processed}
\end{pspicture} 
}
\end{minipage}
\begin{minipage}[c]{2.2in}
\vspace*{15mm}
  \caption{Possible cases regarding the relationship between the coverage of $\llrs[i]$
  and other string positions whose LRs have been calculated}
\label{fig:cases}

\end{minipage}

\end{figure}

{\footnotesize
\begin{lstlisting}
    left  $\leftarrow$ LLRS[i].start; // the left boundary of LLRS[i]
    right $\leftarrow$ LLRS[i].start + LLRS[i].length - 1; // the right boundary of LLRS[i]
\end{lstlisting}
}

\noindent{\bf Case 1:} The coverage of $\llrs[i]$ does not connect to or
overlap with any string positions whose LRs have been calculated.  We
will create a bucket in the $ptr$ array covering the string positions
that are covered by $\llrs[i]$ and set up the corresponding $next$
array entry to be the string position that is right after the coverage
of $\llrs[i]$. The following code shows the case condition and the update
made to the $ptr$ and $next$ arrays.

\smallskip 

\begin{tabular}{l}
{\footnotesize
\begin{lstlisting}
if ptr[left] = -1 and ptr[right] = -1 and (left = 1 or ptr[left-1] = -1)
                  and (right = n or ptr[right+1] = -1)  //case 1
   for j = left...right: ptr[j] $\leftarrow$ i;
   next[i] $\leftarrow$ right + 1;
\end{lstlisting}
}
\end{tabular}

\smallskip 

\noindent{\bf Case 2:} The coverage of $\llrs[i]$ connects to or overlaps
with the right side of a string position area whose LRs have been
calculated. We will extend that area's corresponding $ptr$ array bucket to
the coverage of $\llrs[i]$ and update the corresponding $next$ array entry 
to be the string position that is right after the new bucket.

\smallskip 

\begin{tabular}{l}
{\footnotesize
\begin{lstlisting}
else if (ptr[left] = -1 and left $\neq$ 1 and ptr[left-1] $\neq$ -1) and ptr[right] = -1
                        and (right = n or ptr[right+1] = -1) //case 2.1
   for j = left...right: ptr[j] $\leftarrow$ ptr[left-1];
   next[ptr[left-1]] $\leftarrow$ right + 1;

else if (ptr[left] $\neq$ -1) and ptr[right] = -1
                         and (right = n or ptr[right+1] = -1)  //case 2.2
   for j = next[ptr[left]]...right: ptr[j] $\leftarrow$ ptr[left];
   next[ptr[left]] $\leftarrow$ right + 1;
\end{lstlisting}
}
\end{tabular}

\smallskip

\noindent{\bf Case 3:} The coverage of $\llrs[i]$ connects to or
overlaps with the left side of an existing string position area whose
LRs have been calculated. We will left-extend that area's corresponding $ptr$ array
bucket to the coverage of $\llrs[i]$. We need not to update the
corresponding $next$ array entry, since the string position that is
right after the new $ptr$ bucket does not change.

\smallskip

\begin{tabular}{l}
{\footnotesize
\begin{lstlisting}
else if (ptr[right] = -1 and right $\neq$ n and ptr[right+1] $\neq$ -1) and ptr[left] = -1
                         and (left = 1 or ptr[left-1] = -1) //case 3.1 
   for j = left...right: ptr[j] $\leftarrow$ ptr[right+1];

else if (ptr[right] $\neq$ -1) and ptr[left] = -1 
                          and (left = 1 or ptr[left-1] = -1) //case 3.2
   j $\leftarrow$ left;
   while ptr[j] = -1: ptr[j] $\leftarrow$ ptr[right]; j++;
\end{lstlisting}
}
\end{tabular}

\smallskip

\noindent{\bf Case 4: } Every string position covered by $\llrs[i]$ has
its LR calculated already. In this case, we simply do nothing. 

\smallskip 

\begin{tabular}{l}
{\footnotesize
\begin{lstlisting}
else if ptr[left] $\neq$ -1 and next[ptr[left]] > right: do nothing; //case 4
\end{lstlisting}
}
\end{tabular}

\smallskip

\noindent{Case 5:} The coverage of $\llrs[i]$ bridges two 
string position areas whose LRs have been calculated. 
We will extend the left area's corresponding $ptr$ array bucket up to
the left boundary of the
right area and update the $next$ array entry of the left area to be the
one of the right area. 

\smallskip

\begin{tabular}{l} 
{\footnotesize
\begin{lstlisting}
else
   if ptr[left] = -1: j $\leftarrow$ left; ptr_entry $\leftarrow$ ptr[left-1]; //case 5.1, 5.2
   else:              j $\leftarrow$ next[ptr[left]]; ptr_entry $\leftarrow$ ptr[left]; //case 5.3, 5.4
   while ptr[j] = -1: ptr[j] $\leftarrow$ ptr_entry; j++;
   next[ptr[ptr_entry]] $\leftarrow$ next[ptr[j]];
\end{lstlisting}
}
\end{tabular}

\smallskip

\begin{lemma}[Correctness]
\label{lem:table-correct}
The two-table system's invariance is maintained.
\end{lemma}

\begin{proofsketch}
  Let us call a $ptr$ array bucket $ptr[i\ldots j]$ as a \emph{tail
    bucket} if $j=n$ or $ptr[j+1] = -1$.  (1) We first prove that the
  invariance is maintained on tail buckets.  Observe that any tail
  $ptr$ bucket is created in Case 1 (Figure~\ref{fig:cases}) and can
  be extended in Case 2 as well as in Case 3 if the black bucket in
  Case 3 was also a tail bucket.  The update to the tail bucket as
  well as the corresponding $next$ array entry guarantees that, for
  any $i$ belonging to the coverage of a tail bucket,
  $next\bigl[ptr[i]\bigr]$ is \emph{ideally} equal to the index of the
  string position that is right after the bucket (or $n+1$ if no such
  string position exists). So obviously the invariance is maintained.
  (2) We now prove the invariance is also maintained on non-tail
  buckets.  Observe that any non-tail bucket is created in Case 5 from
  the merge of the left black bucket and the new $\llrs[i]$'s
  coverage. After such non-tail bucket is created, for any position
  $i$ belonging to a non-tail bucket, $next\bigl[ptr[i]\bigr]$ is at
  least as large as the index of the string position that is following
  the right black bucket in Case 5. That means $next\bigl[ptr[i]\bigr]
  - i$ is larger than the size of any unprocessed $llrs$ array
  element. This guarantee is maintained, because every
  $next\bigl[ptr[i]\bigr]$ only monotonically increases.  So, the
  invariance is also maintained for non-tail buckets. (3) Because the
  invariance is well maintained for all $ptr$ buckets, it is safe to
  have the condition checking as we have written for Case 4.
\end{proofsketch}

\begin{lemma}[Time complexity]
\label{lem:table-time}
  The two-table system is maintained using a total of $O(n)$ time
  over the course of the processing of the $\llrs$ array elements. 
\end{lemma}

\begin{proofsketch}
  Observe that the updates to the $ptr$ array are made only to those
  entries whose values were $-1$ and the new values from the updates
  are all positive. So there are no more than $n$ updates
  to the $ptr$ array. It is obvious that the number of updates made to
  the $next$ array is no more than the number of updates made to the
  $ptr$ array. Other than $ptr$ and $next$ array updates, the rest of
  the maintenance work for the two-table system when processing each
  $\llrs$ array element takes $O(1)$ time.  So the total time cost
  in maintaining the two-table system over the course of the
  processing of the whole $\llrs$ array is $O(n)$.
\end{proofsketch}

\subsection{The final $O(n)$ time and space algorithm.}
\label{subsec:final}


\begin{algorithm}[t]
{\small
  \caption{The $O(n)$ time and space algorithm for finding the
    leftmost LR for every non-singleton string position of $S$.}
\label{algo:every-one}
\KwIn{The rank array and the lcp array of the string $S$} 
\KwOut{The leftmost LR covering every non-singleton string position of $S$.}

\smallskip 

\tcc{Calculate the $\llrs$ array using Lemma~\ref{lem:llr}. \newline
Initialize the $LRS$ array  and the auxiliary $ptr$ and $next$
  arrays.}
\For{$i=1,2,\ldots,n$\label{line:init-start}}{
 $\llrs[i] \leftarrow \langle i,\max\{\lcp[\rank[i]], \lcp[\rank[i]+1]\}\rangle$ 
  \tcp*{$\llr_i$, in the format of $\langle start, length \rangle$}
  $\lrs[i] \leftarrow \langle -1,0\rangle$ \tcp*{$\lr_i$, in the format of $\langle start, length \rangle$}
  $ptr[i]\leftarrow -1$; \ \ \   
  $next[i]\leftarrow -1$\label{line:init-end} \;
}

\smallskip 

Stably sort $\llrs[1\ldots n]$ in the descending order of its second
dimension\label{line:sort-2} \tcp*{e.g.: counting sort.}

\bigskip 

\tcc{Find the leftmost LR for every position}

$count \leftarrow 0$ \tcp*{The number of non-singleton string positions that have their
  LRs calculated.} 

\For{$i=1,2,\ldots,n$\label{line:for-2}}{
  \lIf{$count = n$ \emph{or} $\llrs[i].length = 0$}{
    break\label{line:earlystop} \tcp*{Early stop}
  }

\smallskip

    $left  \leftarrow LLRS[i].start$;\ \ 
    $right \leftarrow LLRS[i].start + LLRS[i].length - 1$; \tcp{The boundaries of $LLRS[i]$.}

  \tcc{$first=s_i$ of $P_i=[s_i,e_i]$ if $P_i$ is not empty.}

  \lIf{$ptr[left]=-1$}{
    $first \leftarrow left$;
  }
  \lElse{
    $first \leftarrow next\bigl[ptr[left]\bigr]$\label{line:si}\;
  }

\smallskip

  \lIf{$first > right$}{
    continue \tcp*{Detect the fact that $P_i$ is empty.\label{line:detect}}
  }

\smallskip

\tcc{Calculate the the leftmost LR of every position in $P_i = [s_i,e_i]$.}
   $j \leftarrow first$\label{line:lr-start}\;
   \While{$j\leq right$ \emph{and} $ptr[j]=-1$}{
     $\lrs[j] \leftarrow \langle \llrs[i].start, \llrs[i].length\rangle $;
      $count \leftarrow count + 1$;
      $j\leftarrow j+1$\label{line:lr-end}\;
    }

    \smallskip 
    
    Update the two-table system here using the code presented in Section~\ref{subsec:maintain}\label{line:2-table}.
 }
    \Return{$\lrs[1\ldots n]$}
}
\end{algorithm}


By combining the conceptual Algorithm~\ref{algo:concept}, the
high-level strategy for the fast implementation, and the two-table
system's maintenance mechanism, we are ready to produce the final
$O(n)$ time and space algorithm that can find the leftmost LR of every
string position. Algorithm~\ref{algo:every-one} shows the pseudocode.
It starts with the calculation of the $\llrs$ array and the
initialization of the $\lrs$, $ptr$, and $next$ arrays
(Line~\ref{line:init-start}--\ref{line:init-end}).  It then sorts the
$\llrs$ array in the descending order of the lengths of the array
elements using a linear and stable sorting procedure
(Line~\ref{line:sort-2}). It then uses the {\tt for} loop
(Line~\ref{line:for-2}) to process every $\llrs$ array element with
possible early stop (Line~\ref{line:earlystop}).  Using the two-table
system, the value of $s_i$ is calculated by Line~\ref{line:si} if
$P_i$ is not empty; otherwise, the fact that $P_i$ is empty will also
be detected by Line~\ref{line:detect}. After $s_i$ is calculated,
finding the LR of each position in $P_i$ becomes obvious
(Line~\ref{line:lr-start}--\ref{line:lr-end}). After the LR finding
work is done, we will update the two-table system (Line~\ref{line:2-table}) using
the code presented in Section~\ref{subsec:maintain}.

\begin{lemma}
\label{lem:every-one}  
Given the lcp array and the rank array, Algorithm~\ref{algo:every-one}
calculates the leftmost LR of every non-singleton position of a string $S$
of size $n$ using a total $O(n)$ time and space.
\end{lemma}

\begin{proofsketch}
  (1)\emph{Correctness.} The correctness of Algorithm~\ref{algo:every-one} immediately
  follows from of Lemma~\ref{lem:concept} and
  Lemma~\ref{lem:table-correct}.  (2) All data structures that are
  being involved are the $LCP$, $Rank$, $\llrs$, $\lrs$, $ptr$, and
  $next$ arrays.  Altogether they use $O(n)$ space.  (3) The time cost
  for the initialization
  (Line~\ref{line:init-start}--\ref{line:init-end}) takes $O(n)$ time.
  the stable sorting (Line~\ref{line:sort-2}) uses $O(n)$ time.  The
  rest of the work (Line~\ref{line:for-2}--\ref{line:2-table}) also
  takes $O(n)$ time, because we update every $\lrs$ array element no
  more than once and the two-table system maintenance also takes
  $O(n)$ time (Lemma~\ref{lem:table-time}). So the total time cost
  is $O(n)$.
\end{proofsketch}

\begin{theorem}
\label{thm:every-one}  
Given a string $S$ of size $n$, we can calculate the leftmost LR of every of string position 
using $O(n)$ time and space.   
\end{theorem}

\begin{proofsketch}
  We can construct the suffix array of the string $S$ in a total of
  $O(n)$ time and space using existing algorithms (For ex.,
  \cite{KA-SA2005}).  The rank array is just the inverse suffix array
  and can be directly obtained from SA using $O(n)$ time and
  space. Then we can obtain the lcp array from the suffix array and
  rank array using another $O(n)$ time and space~\cite{KLAAP01}. So
  the total time and space costs for preparing the rank and lcp arrays
  are $O(n)$. The proof of the theorem can then immediately follow
  from  Lemma~\ref{lem:every-one}.
\end{proofsketch}

\remove{

\subsection{Extension: finding all LRs of every string position (REMOVE IT !)}
\label{subsec:ext}

It is possible that a particular position can have multiple LRs.  For
example, if $S={\tt abcabcddbca}$, then $\lr_2$ can be either
$S[1\ldots 3]={\tt abc}$ or $S[2\ldots 4]={\tt bca}$.  However,
Algorithm~\ref{algo:every-one} only returns one of them and resolve
the tie by picking the leftmost one. However, it is easy to modify
Algorithm~\ref{algo:every-one} to return all the LRs of every string
position, without changing the mechanism that maintains the
two-table system.

Idea: 

When processing an $\llrs[i]$,

(1) if its left-end has its lr already calculated in the past, simply
travel from $next\Bigl[ptr\bigl[\llrs[i].start\bigr]\Bigr]$ toward the
left end of $\llrs[i]$ and will stop when seeing a position whose lr's
length is larger than $\llrs[i].length$. Those positions travelled
through will have another LR just calculated as $\llrs[i]$.

(2) similarly handle the right side. 

}


\section{Conclusion}
In this paper, we proposed the problem of finding longest repeats
covering particular string positions, motivated by its applications in
subfields such as computational biology. We proposed optimal
algorithms for finding the (leftmost) longest repeat of every string
position using a total of $O(n)$ time and space based on a novel
two-table system that we designed. We have implemented our
algorithms. Future work can be an
experimental study of the implementation.

\section*{References}

\small 

\bibliographystyle{elsarticle-num}

\bibliography{bibjsv,repeat,pm}

\newpage

\appendix

\section*{Appendix}

\begin{table}[h!]
\center
\def\0{\phantom{0}}
{\footnotesize
\begin{tabular}{c|c|c|l}
\hline 
$i$ & $\lcp[i]$  & $\mathit{\sa}[i]$ & suffixes\\
\hline
\hline
\01 & 0 & 11\0  &{\tt i}\\
\02 & 1 & \08\0  & {\tt  ippi}\\
\03 & 1 & \05\0  & {\tt  issippi}\\
\04 & 4 & \02\0  & {\tt  ississippi}\\
\05 & 0 & \01\0  & {\tt  mississippi}\\
\06 & 0 & 10\0  & {\tt  pi}\\
\07 & 1 &  \09\0  & {\tt ppi}\\
\08 & 0 & \07\0  & {\tt sippi}\\
\09 & 2  & \04\0  & {\tt sissippi}\\
10 & 1  & \06\0  & {\tt ssippi}\\
11 & 3 & \03\0  & {\tt ssissippi}\\
12 & 0 & -- & --\\
\hline
\end{tabular}
}
\bigskip
\caption{The suffix array and the lcp array of an example string $S={\tt mississippi}$.}
\label{tab:suflcp}
\end{table}


\begin{algorithm}[h!]
{\small
  \caption{Find all LRs that cover a given position $k$}
\label{algo:one-all}
\KwIn{The position index $k$, and the rank array and 
      the lcp array of the string $S$} 
\KwOut{All LRs that cover position $k$ or find no such LR.}

\smallskip 

\tcc{Find the length of $\lr_k$.}
$length \leftarrow 0$\;
\For{$i = k$ down to $1$}{
  $L \leftarrow \max\{\lcp[\rank[i]],\lcp[\rank[i]+1]\}$\tcp*{Length
    of $\llr_i$}

  \If(\tcp*[f]{$\llr_i$ does not exist or does not cover $k$.}){$L=0$ or $i+L-1<k$}{break\tcp*{Early stop}}
  \ElseIf{$L\geq length$}
   {$length\leftarrow L$\;}
}


\smallskip 

\tcc{Print all LRs that cover position $k$.}

\If(\tcp*[f]{$\lr_k$ does exist.}){$length>0$}{
  \For{$i = k$ down to $1$}{
    $L \leftarrow \max\{\lcp[\rank[i]],\lcp[\rank[i]+1]\}$\tcp*{Length
      of $\llr_i$}
   \If(\tcp*[f]{$\llr_i$ does not exist or does not cover $k$.}){$L=0$ or $i+L-1<k$}{break\tcp*{Early stop}}
    \ElseIf{$L = length$}
    {Print $\lr_k\leftarrow \langle i,length \rangle$\;}
  }
}
\lElse{Print $\lr_k\leftarrow \langle -1,0\rangle$\tcp*{$\lr_k$ does not exist.}}
}
\end{algorithm}

\end{document}